\documentclass[]{eptcs}
\usepackage{amsmath, amsthm, amssymb, verbatim, mathtools, fancyvrb, varwidth}
\usepackage[all]{xy}
\usepackage{graphicx}
\usepackage{caption}
\usepackage{subcaption}
\usepackage{hhline}

\DeclareFontFamily{OT1}{slmss}{}
\DeclareFontShape{OT1}{slmss}{m}{n}
     {<-8.5> s*[1.1] rm-lmss8
      <8.5-9.5> s*[1.1] rm-lmss9
      <9.5-11> s*[1.1] rm-lmss10
      <11-15.5> s*[1.1] rm-lmss12
      <15.5-> s*[1.1] rm-lmss17
     }{}

\DeclareSymbolFont{sfoperators}{OT1}{slmss}{m}{n}
\DeclareSymbolFontAlphabet{\mathsf}{sfoperators}
\makeatletter
\def\operator@font{\mathgroup\symsfoperators}
\makeatother

\newtheorem{theorem}{Theorem}[section]
\newtheorem{lemma}[theorem]{Lemma}

\newenvironment{definition}[1][Definition.]{\begin{trivlist}
\item[\hskip \labelsep {\bfseries #1}]}{\end{trivlist}}

\newenvironment{qv}
{\quote\Verbatim}
{\endVerbatim\endquote}

\DeclareMathOperator{\id}{id}
\DeclareMathOperator{\proc}{proc}
\DeclareMathOperator{\type}{type}
\DeclareMathOperator{\addr}{addr}
\DeclareMathOperator{\val}{val}
\DeclareMathOperator{\acyclic}{acyclic}
\DeclareMathOperator{\irreflexive}{irreflexive}
\DeclareMathOperator{\po}{po}
\DeclareMathOperator{\pol}{pol}
\DeclareMathOperator{\co}{co}
\DeclareMathOperator{\rf}{rf}
\DeclareMathOperator{\rfe}{rfe}
\DeclareMathOperator{\fr}{fr}
\DeclareMathOperator{\fre}{fre}
\DeclareMathOperator{\corf}{co;rf}
\DeclareMathOperator{\frrf}{fr;rf}
\DeclareMathOperator{\com}{com}
\DeclareMathOperator{\comp}{com^+}
\DeclareMathOperator{\hb}{hb}
\DeclareMathOperator{\ppo}{ppo}
\DeclareMathOperator{\fence}{fence}
\DeclareMathOperator{\prop}{prop}
\DeclareMathOperator{\poto}{\xrightarrow{po}}
\DeclareMathOperator{\polto}{\xrightarrow{pol}}
\DeclareMathOperator{\coto}{\xrightarrow{co}}
\DeclareMathOperator{\rfto}{\xrightarrow{rf}}

\DeclareMathOperator{\frto}{\xrightarrow{fr}}

\DeclareMathOperator{\frrfto}{\xrightarrow{fr;rf}}
\DeclareMathOperator{\corfto}{\xrightarrow{co;rf}}
\DeclareMathOperator{\comto}{\xrightarrow{com}}
\DeclareMathOperator{\compto}{\xrightarrow{com^+}}

\newcommand{\rfinv}{\rf^{-1}}

\title{An ACL2 Mechanization of an Axiomatic Framework for Weak Memory}
\author{Benjamin Selfridge
\institute{University of Texas at Austin\\ Austin, TX}
\email{benself@cs.utexas.edu}
}
\def\titlerunning{Mechanizing Weak Memory}
\def\authorrunning{Benjamin Selfridge}
\begin{document}
\maketitle

\begin{abstract}Proving the correctness of programs written for multiple processors is a challenging problem, due in no small part to the weaker memory guarantees afforded by most modern architectures. In particular, the existence of store buffers means that the programmer can no longer assume that writes to different locations become visible to all processors in the same order. However, all practical architectures do provide a collection of weaker guarantees about memory consistency across processors, which enable the programmer to write provably correct programs in spite of a lack of full sequential consistency. In this work, we present a mechanization in the ACL2 theorem prover of an axiomatic weak memory model (introduced by Alglave et al. \cite{alglave_cats}). In the process, we provide a new proof of an established theorem involving these axioms.
\end{abstract}

\section{Introduction}

Analysis of sequential programs is a well-understood problem for which a variety of proof techniques and methodologies exist. \cite{hoare69} Many of these techniques can be adapted to a multiprocessor setting if we assume \emph{sequential consistency} (SC) - i.e., that for any concurrent execution of the program, there exists an interleaving of the memory events that is consistent with both the program order and the communication dependencies between processes. \cite{lamport79, owicki76} However, sequential consistency turns out to be a much stronger requirement than is practically necessary. Moreover, due to the inherently high runtime and resource penalties of SC, designers of multiprocessor architectures are motivated to relax this constraint in order to achieve better performance.

To understand why a lack of sequential consistency impacts us as programmers, consider the following example. Suppose our architecture consists of a number of processors $P_1, \ldots, P_n$ and a shared memory $M$. Assume that when a processor issues a write to memory, that write is immediately visible to all other processors. 

\begin{figure}[t]
\[
\begin{array}{c|c}
P_0   &   P_1 \\ \hhline{=|=}
x \leftarrow 1   & y \leftarrow 1 \\ 
r_0 \leftarrow y & r_1 \leftarrow x \\ 
\end{array}
\]
\caption{A multiprocessor program execution. The final state $r_0 = 0$, $r_1 = 0$ is prohibited by sequential consistency, but is possible on an architecture with store buffers.}
\label{sb}
\end{figure}

Consider the program execution represented in Figure \ref{sb}. Each processor assigns the value $1$ to memory location $x$ or $y$, and reads the value at the other location into a register. (Assume $x$ and $y$ are both initially equal to $0$.) Now, we ask the question: what are the possible values of registers $r_0$ and $r_1$ after running this program? It is easy to see that $r_0 = 1$, $r_1 = 1$ is one possible final state, obtained by a scheduler that alternates between $P_0$ and $P_1$. We can also obtain $r_0 = 0$, $r_1 = 1$ by running $P_0$'s program to the end, and then subsequently running $P_1$'s program to the end. Likewise, it is also possible to obtain $r_0 = 1$, $r_1 = 0$. These are the only possible final states, because this (sketch of an) architecture is sequentially consistent; every processor completely executes its first instruction before continuing to the second.

Now, consider the following modification of this architecture. Each of the processors $P_i$ is equipped with a \emph{store buffer} $B_i$. When $P_i$ issues a write, instead of propagating the write directly to shared memory, the write is initially sent to buffer $B_i$. That write will eventually hit memory, although we have no guarantee of when that will happen (unless the programmer inserts an explicit memory fence). If $P_i$ wishes to read a value from memory, it first checks its own store buffer to see if it has issued any pending writes to that memory location. If it has, it uses that value; otherwise, it obtains the value from memory. 

If we run the same program on this architecture, it is easy to see that the final state $r_0 = 0$, $r_1 = 0$ is obtainable if neither processor's store buffer is flushed before the reads are performed; both processors issue a write, but those writes are not globally visible by the time each process issues its read, and hence both processors read the ``old'' values of $x$ and $y$. This is a clear violation of sequential consistency. There is no way to linearly order the instructions of the two programs as atomic memory events and obtain this final state; nevertheless, this behavior is possible on this architecture. This odd behavior isn't merely a theoretical possibility; it is actually observable on x86 machines.

In spite of the fact that we do not generally have sequential consistency, most weaker memory models do uphold a set of guarantees which, though they are not as strong as sequential consistency, do prohibit certain behaviors. These guarantees vary greatly from model to model \cite{boudol09, chong08, owens09, sarkar11, sarkar09}, and the variety and abundance of these models suggests the need for a more generic framework for weak memory. Such a framework ought to be both general enough to capture the semantics of all modern architectures, and strong enough to enforce meaningful constraints that are universally upheld. One such framework is introduced in Alglave et al. \cite{alglave_cats}, and in this paper we present its mechanization in ACL2. Furthermore, we present a new proof of an established theorem about this framework, and we discuss the mechanized proof.

A brief notational remark: throughout this paper, given a relation $R$, we will let $R^+$ denote the irreflexive transitive closure of $R$. Given two relations $R$ and $Q$, we let $R;Q$ denote the sequencing of $R$ and $Q$, i.e.
\[
x \xrightarrow{R;Q} y \text{ iff. } \exists p, \text{ } x \xrightarrow{R} p \xrightarrow{Q} y.
\]

\section{Background: An Axiomatic Framework for Weak Memory}

The execution of a sequential program results in a linear sequence of events (usually reads or writes from/to a location in memory). The event order derived from this sequence is called the \emph{program order}. The program order is a total order on all events, and from this order we can reason in a straightforward way about the possible final states that can result from a run of the program by considering all possible event orderings and demonstrating that they all produce a final state in a particular configuration.

With concurrent programs, however, the situation is more complicated. Generally speaking, an execution on a concurrent machine is not simply a sequence of events with a global program order. Events that occur on different processors are not necessarily comparable, because a write issued by one processor may not be visible to any other processors for some time (despite being immediately visible to the process that executed it). Therefore, in order to specify a set of requirements for our weaker memory guarantees, we need a weakened definition of a program execution that retains enough structure to be amenable to subsequent constraints and analyses. In this section, we describe a compelling axiomatic framework for weak memory \cite{alglave_cats}, which includes both a more general notion of execution for multiple processors and a parameterized set of requirements that is meant to characterize all modern multiprocessor architectures.

\subsection{Concurrent Executions}

We begin with two definitions.

\begin{definition}
An \emph{event} $e$ is an object which consists of a unique identifier $\id(e)$, a process $\proc(e)$, a type $\type(e)$ which identifies $e$ as being either a read or a write, an address $\addr(e)$ equal to the address in memory that $e$ reads from or writes to, and a value $\val(e)$ equal to the value read or written by $e$.
\end{definition}

\begin{definition}
An \emph{execution} is a tuple $E = (\mathbb{E}, \po, \co, \rf)$ where $\mathbb{E}$ is a collection of events, and $\po$, $\co$, and $\rf$ are all relations on $\mathbb{E}$ satisfying:
\begin{itemize}
  \item $\po$ is a total order on events, when restricted to a single process
  \item $\co$ is a total order on writes, when restricted to a single address
  \item $\rf$ is a relation from writes to reads such that for all reads $r \in \mathbb{E}$, there exists a unique write $w \in \mathbb{E}$ such that $w \rfto r$ (we also require that $\val(w) = \val(r)$).
\end{itemize}
The relation $\po$ is undefined on events belonging to different processes, and likewise, $\co$ is undefined on any pair of events that are not writes to the same address.
\end{definition}

The relation $\po$ is our concurrent version of program order; it is a total order not on all events, but only on those belonging to the same processor. The ``coherence order'' $\co$ is a total order on writes to the same location in memory. This order corresponds to our intuition that the writes to each individual location hit memory in a particular sequential order. The read-from relation $\rf$ captures the dependency between writes and reads; $w \rfto r$ means ``$r$ takes its value from the write $w$.'' \footnote{The reader may be wondering why we choose to write $w \rfto r$ rather than $r \rfto w$ - the latter certainly seems more sensible when read aloud (``r read-from w''). The reason is that the direction of the arrow is meant to represent a dependency between two events, with the arrow pointing toward the dependent (``later'') event. This will enable us to state our weak memory requirements as assertions of the acyclicity of various combinations of these and other relations.} It is a surjective relation with a one-sided inverse function, $\rfinv$.

The purpose of $\co$ and $\rf$ is to capture interprocess dependencies between events occurring at the same location; $\co$ captures dependencies between two writes arising from their relative visibility with respect to time, and $\rf$ captures the dependency of reads on the writes they take their value from. However, it is also intuitively possible to have a write ``depend'' on a read. If $w, w'$ are writes and $r$ is a read such that $w' \rfto r$ and $w' \coto w$, then there is a sense in which $w$ ``comes after'' $r$, because $r$ takes its value from an earlier write. Therefore, we have another relation, which we refer to as the ``from-read'' relation.

\begin{definition}
Let $E = (\mathbb{E}, \po, \co, \rf)$ be an execution. The ``from-read'' relation $\fr$ is defined as
\[
\fr = \rfinv ; \co,
\]
i.e. $r \frto w$ if there exists a write $w'$ such that $w' \rfto r$ and $w' \coto w$. (Note that this is equivalent to stating that $\rfinv(r) \coto w$.)
\end{definition}

Our three relations $\rf$, $\co$, and $\fr$ will be sufficient to specify certain communication dependencies regarding reads and writes to the same location. We abbreviate the three into a single relation.

\begin{definition}
Let $E = (\mathbb{E}, \po, \co, \rf)$ be an execution. The relation $\com$ is defined as
\[
\com = \co \cup \rf \cup \fr,
\]
i.e. $x \comto y$ if $x \coto y$, $x \rfto y$, or $x \frto y$.
\end{definition}

The $\po$ and $\com$ relations represent two distinct types of dependencies between events; $\po$ captures \emph{per-process} dependencies, and $\com$ relation captures \emph{per-location} dependencies. The existence of these two relations suggests two distinct views of our event graph. The first is the per-process view, where we organize all the events by the process they belong to, and list them in program order (see Figure \ref{events:process}). The second is the per-location view, where we organize the events by the memory location at which they occur, and list each write event in coherence order (see Figure \ref{events:location} for an example of what this might look like for a particular location $M_0$).
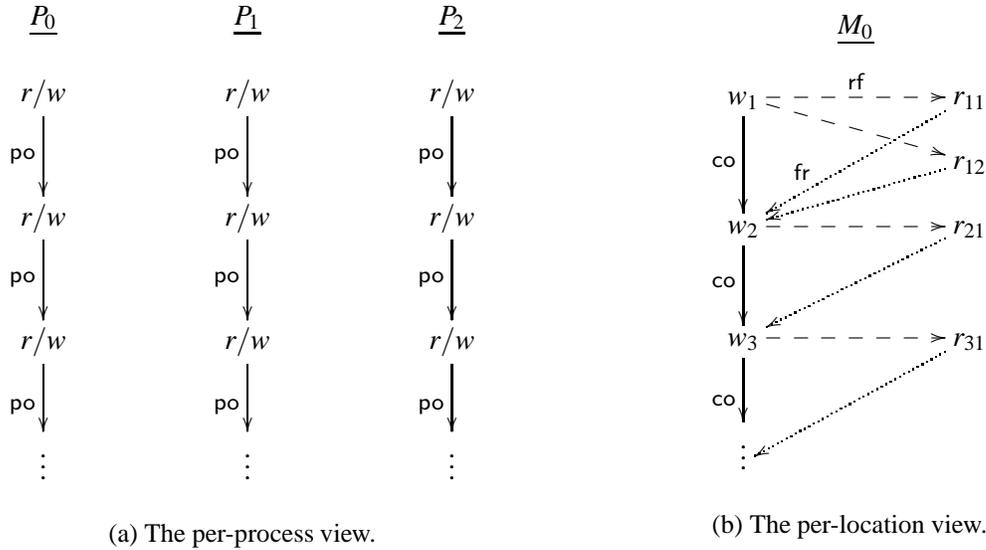
\begin{figure}[t]
\begin{subfigure}{.5\textwidth}
\[
\xymatrix@R1pc{
\underline{P_0} && \underline{P_1} && \underline{P_2}\\
r/w \ar[dd]_{\po} && r/w \ar[dd]_{\po} && r/w \ar[dd]_{\po} \\\\
r/w \ar[dd]_{\po} && r/w \ar[dd]_{\po} && r/w \ar[dd]_{\po} \\\\
r/w \ar[dd]_{\po} && r/w \ar[dd]_{\po} && r/w \ar[dd]_{\po} \\\\
\vdots && \vdots && \vdots \\
}
\]
\caption{The per-process view.}
\label{events:process}
\end{subfigure}
\begin{subfigure}{.5\textwidth}
\[
\xymatrix@R1pc{
& \underline{M_0} &\\
w_1 \ar[dd]_{\co} \ar@{-->}[rr]^{\rf} \ar@{-->}[drr] && r_{11} \ar@{..>}[ddll]_(0.7){\fr} \\
                                                     && r_{12} \ar@{..>}[dll] \\
w_2 \ar[dd]_{\co} \ar@{-->}[rr]                      && r_{21} \ar@{..>}[ddll]\\\\
w_3 \ar[dd]_{\co} \ar@{-->}[rr]                      && r_{31} \ar@{..>}[ddll]\\\\
\vdots 
}
\]
\caption{The per-location view.}
\label{events:location}
\end{subfigure}
\caption{Two views of memory events. In figure (b), solid lines are $\co$, dashed lines are $\rf$, and dotted lines are $\fr$. For $\po$, $\co$ and $\fr$, not all arrows are pictured, as $\po$ and $\co$ are transitively closed.}
\label{events}
\end{figure}

\subsection{Sequential Consistency and SC-Per-Location}\label{sc-per-location-section}

In the previous section, we presented a generalization of the notion of a sequential execution to an arbitrary number of processors. Whereas a sequential execution has a single relation, the program order (which is a total order on all events), a concurrent execution consists of two: its per-process program order $\po$, and the communication dependency relation $\com$. In our framework, the usual definition of sequential consistency \cite{lamport79} is that there exists a completion of the relation $\po \cup \com$ which is a total order on all events. An equivalent way to state this is that the relation $\po \cup \com$ is acyclic, and so we have the following definition:
\begin{definition}
An execution $E = (\mathbb{E}, \po, \co, \rf)$ is \emph{sequentially consistent} (SC) if 
\[
\acyclic(\po \cup \com),
\]
i.e. the union of the $\po$ and $\com$ relations is acyclic.
\end{definition}
As we have already discussed, sequential consistency does not hold in general for modern multiprocessor architectures. However, if we restrict the program order $\po$ to events at the same location, then we get a new, weaker property. As it happens, this property holds for all modern architectures.

To this end, we define another relation, $\pol$, which is the restriction of $\po$ to events that occur at the same location.

\begin{definition}
Let $E = (\mathbb{E}, \po, \co, \rf)$ be an execution. The relation $\pol$ is defined as
\[
\pol = \{ (x,y) \in \mathbb{E} \times \mathbb{E} \mid x \poto y \text{ and } \addr(x) = \addr(y)\},
\]
i.e. $x \polto y$ if $x \poto y$ and $x$ and $y$ have the same address.
\end{definition}

We are now in a position to reproduce the definition for a weakened version of sequential consistency for concurrent executions (originally given in \cite{alglave_cats}), which we refer to as sequential consistency per location.

\begin{definition}
An execution $E = (\mathbb{E}, \po, \co, \rf)$ is \emph{sequentially consistent per location} (SC-Per-Location) if 
\[
\acyclic(\pol \cup \com),
\]
i.e. the union of the $\pol$ and $\com$ relations is acyclic.
\end{definition}

The intuition behind this definition is that if we restrict ourselves to examining one memory location, the system appears to be sequentially consistent. The acyclicity of program order and the communication relations $\co$, $\rf$ and $\fr$ guarantee the existence of a sequential execution of these events that produces the same behavior (for \emph{this} memory location) as the concurrent one. However, this cannot necessarily be generalized to multiple memory locations; the sequential ordering of events for one location may conflict (i.e. create a cycle) with the sequential ordering for another location.

\subsection{The full set of requirements}

SC-Per-Location is one of the four requirements of this framework. It is the only requirement described solely in terms of executions; the other three are defined in terms of a particular architecture. This requires a formal definition of an architecture.

\begin{definition}
An \emph{architecture} is a function $\mathcal{A}$ which maps executions $E = (\mathbb{E}, \po, \co, \rf)$ to tuples 
\[
(\ppo, \fence, \prop)
\] 
such that for all executions $E$,
\begin{itemize}
  \item $\ppo \subseteq \po$
  \item $\fence$ is some relation on events
  \item $\prop$ is some relation on the writes of $\mathbb{E}$ (not necessarily to the same location)
\end{itemize}
\end{definition}
Here, the relation $\ppo$ (``preserved program order'') refers to some subset of the program order that relates events which aren't allowed to be reordered in an execution, $\fence$ refers to pairs of events which are separated by a fence, and $\prop$ (``propagation order'') refers to additional constraints (beyond those specified by $\co$) on the order in which events get propagated to memory.

This definition formulates the notion of an architecture as a set of further restrictions on executions. Depending on how we define the orders $\ppo$, $\fence$, and $\prop$ on an execution, our model will satisfy different memory constraints, because our constraints are defined in terms of these relations. The set of all possible architectures that can be specified from this framework corresponds to all the different ways we can define these relations in terms of a given execution.

The full set of weak memory requirements is as follows. Let $\mathcal{A}$ be an architecture. Then for any execution $E = (\mathbb{E}, \po, \co, \rf)$, we require
\[
\begin{aligned}
  \text{(SC-Per-Location)  } &&& \acyclic(\pol \cup \co \cup \rf \cup \fr)\\
  \text{(No Thin Air)  } &&& \acyclic(\hb)\\
  \text{(Observation)  } &&& \irreflexive(\fre ; \prop; \hb^*)\\
  \text{(Propagation)  } &&& \acyclic(\co \cup \prop)
\end{aligned}
\]
where 
\[
\hb = \ppo \cup \fence \cup \rfe,
\] 
\[
\rfe = \{(x,y) \mid x \rfto y \text{ and } \proc(x) \neq \proc(y)\},
\] 
and 
\[
\fre = \{(x,y) \mid x \frto y \text{ and } \proc(x) \neq \proc(y)\},
\] 
and $\hb^*$ is the reflexive transitive closure of $\hb$.

SC-Per-Location was described above; the other three requirements are discussed thoroughly in \cite{alglave_cats}, and are best understood in the context of the various examples provided in that work. We present the full framework here for completeness, but our investigation into these properties was limited to SC-Per-Location.

\subsection{SC-Per-Location: an alternate definition}

The definition we have for SC-Per-Location makes intuitive sense - it corresponds directly to the classic definition of sequential consistency. However, as it turns out, this definition is equivalent to a seemingly weaker property (originally introduced in \cite{alglave_thesis}), which we reproduce below. 

\begin{definition}
An execution $E = (\mathbb{E}, \po, \co, \rf)$ satisfies the property \emph{SC-Per-Location-2} if 
\[
\forall x, y \in \mathbb{E}, \text{ } x \polto y \implies \neg (y \compto x)
\]
i.e. no two events be related by $\pol$ in one direction and $\comp$ in the other direction.
\end{definition}

This alternate definition captures the intuition that if an event precedes another event in program order, it cannot have a communication dependency (or a sequence of dependencies) on the latter event. Clearly, the existence of such a dependency would create a cycle in $\pol \cup \com$, and so it is easy to see that SC-Per-Location implies SC-Per-Location-2. As it turns out, this definition of SC-Per-Location-2 is actually equivalent to the one given in Section \ref{sc-per-location-section}; this was first proved in Alglave \cite{alglave_thesis} and we give a new proof of this result in the next section.

Now, as it turns out, the $\comp$ relation can be written as the union of the five relations $\rf, \co, \fr, \corf,$ and $\frrf$. We state this as a theorem, and provide a sketch of the proof.

\begin{theorem}\label{rewrite_com+}
Let $E = (\mathbb{E}, \po, \co, \rf)$ be an execution. Then we have
\[
\comp = \com \cup (\corf ) \cup (\frrf).
\]
\end{theorem}
\begin{proof}
Suppose we have a path $x \to p_1 \to \cdots \to p_k \to y$, where $\to$ abbreviates $\comto$. We proceed by induction on $k$. If $k = 0$, we have $x \comto y$, and we are done.

Now, suppose $k \geq 1$ and assume inductively that the theorem holds for the all shorter paths. We have
\[
x \to p_1 \to \cdots \to p_k \to y.
\]
Now, the path $p_1 \to \cdots \to p_k \to y$ is a shorter path, and hence by our induction hypothesis, we have $p_1 \comto y$, $p_1 \corfto y$, or $p_1 \frrfto y$. Furthermore, we have $x \coto p_1$, $x \rfto p_1$, or $x \frto p_1$. If we consider all these cases (many of which are vacuous due to the fact that $\co$, $\rf$ and $\fr$ all relate events of specific types), it is easy to demonstrate that $x \comto y$, $x \corfto y$, or $x \frrfto y$.
\end{proof}

From this theorem, we can clearly see that an execution satisfies SC-Per-Location-2 if and only if it does not contain any of the patterns in Figure \ref{patterns}. We will ultimately prove that SC-Per-Location is equivalent to SC-Per-Location-2, which guarantees that a cycle of any kind in $\pol \cup \com$, no matter how big the cycle is, will imply the existence of a ``mini''-cycle of one of these five variants.

\begin{figure}
\[
\xymatrix{
w_1 \ar@/_0.5pc/[d]_{\pol} && r \ar@/_0.5pc/[d]_{\pol} && w \ar@/_0.5pc/[d]_{\pol} \\
w_2 \ar@/_0.5pc/[u]_{\co} && w \ar@/_0.5pc/[u]_{\rf} && r \ar@/_0.5pc/[u]_{\fr}
}
\]

\[
\xymatrix{
r \ar[d]_{\pol} & w_2 \ar[l]_{\rf} && r_1 \ar[d]_{\pol} & w \ar[l]_{\rf} \\
w_1 \ar[ur]_{\co} &&& r_2 \ar[ur]_{\fr}
}
\]
\caption{The five patterns prohibited by SC-Per-Location-2.}
\label{patterns}
\end{figure}
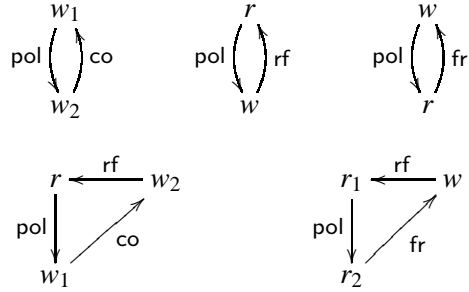

\subsection{An equivalence theorem}

Before we state and prove the equivalence theorem (originally proved in \cite{alglave_thesis}, but proved here in a somewhat more straightforward manner), we first establish two simple lemmas.
\begin{lemma}\label{com-acyclic}
The relation $\comp$ is irreflexive.
\end{lemma}
\begin{proof}
Suppose $x \compto x$. By Theorem \ref{rewrite_com+}, we have three cases.

\emph{Case 1}: $x \comto x$. This is impossible; $\co$ is irreflexive by definition (it is an irreflexive total order), and $\rf$ and $\fr$ are both trivially irreflexive because they only relate events of different types.

\emph{Case 2}: $x \corfto x$. This is impossible; $\corf$ relates writes to reads, and hence is irreflexive.

\emph{Case 3}: $x \frrfto x$. Then there exists an event $z$ with $x \frto z \rfto x$; this in turn implies the existence of an event $y$ with $y \rfto x$, $y \coto z$, and $z \rfto x$. By the uniqueness of writes for the $\rf$ relation, we must have $y = z$; therefore $y \coto y$, which is impossible since $\co$ is irreflexive.
\end{proof}

Upon examination of Figure \ref{events:location}, it is intuitively clear that any two events in this picture either on the same ``level'', or there is a path from one to the other. This is precisely what Lemma \ref{comp-total} says.

\begin{lemma}\label{comp-total}
Let $E = (\mathbb{E}, \po, \co, \rf)$ be an execution, and let $x, y \in \mathbb{E}$ with $\addr(x) = \addr(y)$. Then one of the following holds:
\begin{enumerate}
  \item $x \compto y$
  \item $x$ and $y$ are both writes, and $x = y$
  \item $x$ and $y$ are both reads, and $\rfinv(x) = \rfinv(y)$
  \item $y \compto x$.
\end{enumerate}
\end{lemma}
\begin{proof}
We have four cases, corresponding to $x$ and $y$ each being either reads or writes; however, the symmetry of the read-write cases reduces the number to three. In all three cases, the theorem reduces to the totality of $\co$.

\emph{Case 1}: $x$ is a write, $y$ is a write. Then by totality of $\co$, either $x \coto y$, $y \coto x$, or $x = y$.

\emph{Case 2}: $x$ is a write, $y$ is a read. Then by totality of $\co$, either $x \coto \rfinv(y)$, $x = \rfinv(y)$, or $\rfinv(y) \coto x$. In the first case, $x \corfto y$; in the second, $x \rfto y$; and in the third, $y \frto x$.

\emph{Case 3}: $x$ is a read, $y$ is a read. Then by totality of $\co$, either $\rfinv(x) \coto \rfinv(y)$, $\rfinv(x) = \rfinv(y)$, or $\rfinv(y) \coto \rfinv(x)$. In the first case, $x \frrfto y$; in the second, we are done; and in the third, $y \frrfto x$.
\end{proof}

\begin{theorem}
Let $E$ be an execution. Then $E$ satisfies SC-Per-Location if and only if $E$ satisfies SC-Per-Location-2.
\end{theorem}

%\noindent It is obvious from the definitions (and from the visual depiction of the patterns prohibited by SC-Per-Location-2) that SC-Per-Location implies SC-Per-Location-2. The other direction is not as easy to see; nevertheless, it holds, and we give a new proof of this below.

\begin{proof}
It is clear that SC-Per-Location implies SC-Per-Location-2.

We prove the other direction by contrapositive. Suppose SC-Per-Location does not hold; that is, there exists a cycle in $\pol \cup \com$. Clearly any such cycle is also a cycle in $\pol \cup \comp$ (since $\com \subseteq \comp$). We proceed by induction on the length of this cycle, noting trivially that the length cannot be $1$ (because we know that $\pol$ and $\com$ are both irreflexive).

If the cycle has length two, we must either have $x \polto p \compto x$ or $x \compto p \polto x$, because both of these relations are by themselves acyclic. In either case, the SC-Per-Location-2 condition is clearly violated by $x$ and $p$.

Suppose the cycle has length three or more, i.e.
\[
x \to p_1 \to p_2 \to \cdots \to x,
\]
where $\to$ abbreviates the union of $\pol$ and $\comp$. Also, inductively assume that the existence of a shorter cycle implies that SC-Per-Location-2 does not hold. Assume that $x \compto p_1 \polto p_2$ or $x \polto p_1 \compto p_2$, because otherwise it is clear by transitivity of $\comp$ and $\pol$ that we can obtain a shorter cycle $x \to p_2 \to \cdots \to x$, and so by our inductive hypothesis SC-Per-Location-2 doesn't hold. Then we have several cases, based on Lemma \ref{comp-total}.

\emph{Case 1}: $x \compto p_2$. Then we have the shorter cycle $x \compto p_2 \to \cdots \to x$, and so by our inductive hypothesis, SC-Per-Location-2 does not hold.

\emph{Case 2}: $x$ and $p_2$ are writes where $x = p_2$. Then clearly $x = p_2 \to \cdots \to x$ is a shorter cycle, so by our inductive hypothesis, SC-Per-Location-2 does not hold.

\emph{Case 3a}: $x$ and $p_2$ are reads where $\rfinv(x) = \rfinv(p_2)$, and $x \compto p_1 \polto p_2$. Then it is straightforward to show that $p_2 \compto p_1$, giving $p_1 \polto p_2 \compto p_1$, which violates SC-Per-Location-2.

\emph{Case 3b}: $x$ and $p_2$ are reads where $\rfinv(x) = \rfinv(p_2)$, and $x \polto p_1 \compto p_2$. Then it is straightforward to show that $p_1 \compto x$, giving $x \polto p_1 \compto x$, which violates SC-Per-Location-2.

\emph{Case 4a}: $p_2 \compto x$, and $x \compto p_1 \polto p_2$. Then clearly $p_2 \compto p_1$, giving $p_1 \polto p_2 \compto p_1$, which violates SC-Per-Location-2.

\emph{Case 4b}: $p_2 \compto x$, and $x \polto p_1 \compto p_2$. Then clearly $p_1 \compto x$, giving $x \polto p_1 \compto x$, which violates SC-Per-Location-2.

By Lemma \ref{comp-total} there are no other possibilities. Therefore by induction, if SC-Per-Location does not hold then SC-Per-Location-2 does not hold, and the proof is complete.
\end{proof}

We believe this proof is new. Its direct use of an inductive argument and a ``totality'' lemma (Lemma \ref{comp-total}) for $\comp$ both distinguishes it from the original \cite{alglave_thesis}, and makes its mechanization in ACL2 much easier. One of ACL2's big strengths is its ability to prove theorems inductively, and by understanding an inductive hand proof of this theorem, we were able to make the ACL2 proof much more straightforward.

\section{ACL2 Mechanization}

In this section we present our ACL2 mechanization of the framework and proofs presented above. We make extensive use of the \verb|defun-sk| construct; our definitions of the relations $\po$, $\co$, $\rf$, and $\fr$, as well as various combinations of these relations, are introduced with \verb|defun-sk| in order to make the concepts as general as possible; instead of defining them in terms of a specific data structure (like a graph), we define them as completely general relations which satisfy only the properties we require.

For clarity, we have chosen to present the ACL2 mechanization in a separate section from the preceding one. We have also opted to reproduce most of the definitions, theorems, and even a few key lemmas in order to give the reader a fuller understanding of how these ideas were mechanized. The interested reader might gain some insight into reading the ACL2 code carefully, but is encouraged to skim through it if necessary.

\subsection{Mechanization of Concurrent Executions}

We formalize the concepts of events, $\po$, $\co$, and $\rf$ as constrained functions that satisfy the requirements given in the previous section.
\begin{qv}
(encapsulate
 (((writep *) => *)
  ((readp *) => *)

  ((addr *) => *)
  ((proc *) => *)

  ((po * *) => *)
  ((rf * *) => *)
  ((co * *) => *)
  
  ((rf-inv-fn *) => *))
  
  ; ... constraints omitted
)
\end{qv}
The required properties of these functions are guaranteed by a number of exported theorems, such as totality of \verb|po| on events in the same process, totality of \verb|co| on writes to the same location, and the one-sided invertibility of \verb|rf| (this last property implicitly make use of \verb|rf|'s inverse function \verb|rf-inv-fn|).

We define the function \verb|fr| in terms of \verb|co| and \verb|rf| using ACL2's \verb|defun-sk| construct:
\begin{qv}
(defun-sk fr (x z)
  (exists y
    (and (rf y x) (co y z))))
\end{qv}
We define the ACL2 analogues of sequenced relations $\corf$ and $\frrf$ similarly:
\begin{qv}
(defun-sk co->rf (x z)
  (exists y
    (and (co x y) (rf y z))))
(defun-sk fr->rf (x z)
  (exists y
    (and (fr x y) (rf y z))))
\end{qv}
We define the functions \verb|com| and \verb|pol| as expected:
\begin{qv}
(defun com (x y)
  (or (co x y)
      (rf x y)
      (fr x y)))
(defun pol (x y)
  (and (po x y)
       (equal (addr x) (addr y))))
\end{qv}
The transitive closure of \verb|com| is defined in terms of the existence of a path:
\begin{qv}
(defun com-pathp (path x y)
  (cond ((endp path) (com x y))
        (t (and (com x (car path))
                (com-pathp (cdr path) (car path) y)))))
(defun-sk com+ (x y)
  (exists path (com-pathp path x y)))
\end{qv}
The variable \verb|path| represents the elements between (and not including) \verb|x| and \verb|y|. We prove that we can rewrite \verb|com+| according to Theorem \ref{rewrite_com+}:
\begin{qv}
(defthm rewrite-com+
  (equal (com+ x y)
         (or (com x y)
             (co->rf x y)
             (fr->rf x y))))
\end{qv}
We prove that \verb|com+| is irreflexive, corresponding to Lemma \ref{com-acyclic}:
\begin{qv}
(defthm com+-irreflexive
  (not (com+ x x)))
\end{qv}
And we prove a theorem about the ``totality'' of \verb|com+|, corresponding to Lemma \ref{comp-total}:
\begin{qv}
(defthm com+-totality
  (implies (and (or (readp x) (writep x))
                (or (readp y) (writep y))
                (equal (addr x) (addr y))
                (not (com+ x y))
                (not (and (writep x)
                          (writep y)
                          (equal x y)))
                (not (and (readp x)
                          (readp y)
                          (equal (rf-inv-fn x) (rf-inv-fn y)))))
           (com+ y x)))
\end{qv}

The majority of these theorems were proven by ACL2 with no hints other than the occasional instantiation of witness functions and the selective enabling/disabling of functions and theorems.

\subsection{Mechanization of both definitions of SC-Per-Location}

In order to define SC-Per-Location in ACL2, we need to define the notion of a ``cycle'' in the union of \verb|pol| and \verb|com|. We first define the union of these two relations:
\begin{qv}
(defun pol-com (x y)
  (or (pol x y)
      (com x y)))
\end{qv}
Then we define the notion of a path in $\verb|pol-com|$:
\begin{qv}
(defun pol-com-pathp (path x y)
  (cond ((endp path) (pol-com x y))
        (t (and (pol-com x (car path))
                (pol-com-pathp (cdr path) (car path) y)))))
\end{qv}
If \verb|path| is \verb|nil|, this definition reduces to \verb|(pol-com x y)|. Now, we can define a cycle in \verb|pol-com| as
\begin{qv}
(defun pol-com-cyclep (cycle x)
  (pol-com-pathp cycle x x))
\end{qv}
SC-Per-Location states that there does not exist a cycle in \verb|pol-com|. This can be stated as 
\[
(\forall \verb|x|, \verb|cycle|) \text{  } \verb|(not (pol-com-cyclep cycle x))|.
\]
We can thus define SC-Per-Location in ACL2 as
\begin{qv}
(defun-sk sc-per-location-1 ()
  (forall (x cycle)
          (not (pol-com-cyclep cycle x))))
\end{qv}
SC-Per-Location-2 can be easily defined as
\begin{qv}
(defun-sk sc-per-location-2 ()
  (forall (x y)
          (implies (pol x y)
                   (not (com+ y x)))))
\end{qv}

\subsection{Mechanization of the equivalence proof, Part 1}
As before, the easy part of the equivalence proof is the fact that \verb|(sc-per-location-1)| implies \linebreak\verb|(sc-per-location-2)|. The first step involved proving an unquantified version of the theorem, where we assume \verb|(pol x y)| and \verb|(com+ y x)|, and consider the three cases afforded by \verb|rewrite-com+|:
\begin{qv}
(defthm pol-com-cycle
  (implies (and (pol x y)
                (com y x))
           (pol-com-cyclep (list y) x)))
(defthm pol-co->rf-cycle
  (implies (and (pol x y)
                (co->rf y x))
           (pol-com-cyclep (list y (co->rf-witness y x)) x)))
(defthm pol-fr->rf-cycle
  (implies (and (pol x y)
                (fr->rf y x))
           (pol-com-cyclep (list y (fr->rf-witness y x)) x)))
\end{qv}
Then we add \verb|sc-per-location-1| back into these theorems with \verb|:instance| hints:
\begin{qv}
(defthm pol-com-not-sc-per-location-1
  (implies (and (sc-per-location-1)
                (pol x y))
           (not (com y x)))
  :hints (("Goal"
           :use ((:instance sc-per-location-1-necc
                            (x x)
                            (potential-cycle (list y)))))))
(defthm pol-co->rf-not-sc-per-location-1
  (implies (and (sc-per-location-1)
                (pol x y))
           (not (co->rf y x)))
  :hints (("Goal"
           :use ((:instance sc-per-location-1-necc
                  (x x)
                  (potential-cycle (list y (co->rf-witness y x))))))))
(defthm pol-fr->rf-not-sc-per-location-1
  (implies (and (sc-per-location-1)
                (pol x y))
           (not (fr->rf y x)))
  :hints (("Goal"
           :use ((:instance sc-per-location-1-necc
                  (x x)
                  (potential-cycle (list y (fr->rf-witness y x))))))))
\end{qv}
Finally, we state the fully quantified version of the theorem, which ACL2 proves immediately:
\begin{qv}
(defthm sc-per-location-1-implies-2
  (implies (sc-per-location-1)
           (sc-per-location-2)))
\end{qv}

\subsection{Mechanization of the equivalence proof, Part 2}
The proof that \verb|(sc-per-location-2)| implies \verb|(sc-per-location-1)| was broken down into 4 steps:
\begin{enumerate}
  \item Prove that any 2-cycle in \verb|pol-com+| violates \verb|sc-per-location-2|, and that if there is a cycle of length 3 or greater in \verb|pol-com+|, where \verb|pol-com+| is the union of \verb|pol| and \verb|com+|, then there is a smaller cycle in \verb|pol-com+|, and   \item Use the theorem in step 1 to define a function, \verb|collapse-cycle|, which takes a cycle in \verb|pol-com+| and produces a pair \verb|(x y)| such that \verb|(pol x y)| and \verb|(com+ y x)|
  \item Combine steps 1 and 2 to show that if we have a cycle in \verb|pol-com| (i.e. a violation of \linebreak\verb|sc-per-location-1|), we have a pair \verb|(x y)| which violates \verb|sc-per-location-2|
\end{enumerate}

Step 1 is summarized by two theorems, one that states that 2-cycles in \verb|pol-com+| violate \linebreak\verb|sc-per-location-2|, and one that takes cycles longer than 2 and produces a smaller cycle.
\begin{qv}
(defthm cycle-2
  (implies (and (pol-com+-cyclep cycle x) 
                (endp (cdr cycle))
                (not (and (pol x (car cycle))
                          (com+ (car cycle) x))))
           (and (pol (car cycle) x)
                (com+ x (car cycle)))))
(defthm collapse-cycle-thm
  (implies (and (not (pol-com+-cyclep (list p1) x))
                (not (pol-com+-cyclep (list* p2 rst) x))
                (not (pol-com+-cyclep rst x))
                (not (pol-com+-cyclep (list p2) p1)))
           (not (pol-com+-cyclep (list* p1 p2 rst) x))
  :hints (("Goal"
           :cases ((com+ x p2)
                   (and (writep x)
                        (writep p2)
                        (equal x p2))
                   (and (readp x) 
                        (readp p2)
                        (equal (rf-inv-fn x) (rf-inv-fn p2)))
                   (com+ p2 x)))))
\end{qv}
Notice that the case split corresponds exactly to Theorem \ref{comp-total}, just as in the written proof.

For Step 2, we define the function \verb|collapse-cycle| to shorten the cycle according to the previous theorem. The \verb|collapse-cycle| function satisfies the property that if it is given a violation of \linebreak\verb|sc-per-location-1|, it produces a violation of \verb|sc-per-location-2|:
\begin{qv}
(defun collapse-cycle (cycle x)
  (let* ((p1 (car cycle))
         (p2 (cadr cycle))
         (rst (cddr cycle)))
    (cond ((endp cycle) (mv nil x))
          ((endp (cdr cycle))
           (if (pol x (car cycle))
               (mv x (car cycle))
             (mv (car cycle) x)))
          ((pol-com+-cyclep (list* p2 rst) x)
           (collapse-cycle (list* p2 rst) x))
          ((pol-com+-cyclep rst x)
           (collapse-cycle rst x))
          ((pol-com+-cyclep (list p2) p1)
           (collapse-cycle (list p2) p1))
          (t (collapse-cycle (list p1) x)))))
(defthm collapse-cycle-pol-com+
  (implies (pol-com+-cyclep cycle x)
           (mv-let (new-x new-y)
                   (collapse-cycle cycle x)
                   (and (pol new-x new-y)
                        (com+ new-y new-x)))))
\end{qv}

For Step 3, we first add in the quantifier for \verb|sc-per-location-2|:
\begin{qv}
(defthm sc-per-location-1-implies-2-unquantified
  (implies (sc-per-location-2)
           (not (pol-com-cyclep cycle a)))
  :hints (("Goal"
           :use ((:instance sc-per-location-2-necc
                            (x (mv-let (new-x new-y)
                                       (collapse-cycle cycle a)
                                       (declare (ignore new-y))
                                       new-x))
                            (y (mv-let (new-x new-y)
                                       (collapse-cycle cycle a)
                                       (declare (ignore new-x))
                                       new-y)))))))
\end{qv}
The result follows immediately:
\begin{qv}
(defthm sc-per-location-2-implies-1
  (implies (sc-per-location-2)
           (sc-per-location-1)))
\end{qv}

\subsection{Mechanizing the other requirements}
The other requirements of this framework were also mechanized in ACL2, using constrained functions to represent $\ppo$, $\fence$, and $\prop$, and with $\rfe$, $\fre$, and $\hb$ defined in terms of these constrained functions. The concepts of No Thin Air, Observation, and Propagation were defined as follows:
\begin{verbatim}
(defun-sk no-thin-air ()
  (forall (x potential-cycle)
          (not (hb-cyclep potential-cycle x))))
(defun-sk observation ()
  (forall x
          (not (fre->prop->hb* x x))))
(defun-sk propagation ()
  (forall (x cycle)
          (not (co-prop-cyclep cycle x))))
\end{verbatim}
We did not investigate these requirements to the extent that we analyzed SC-Per-Location. We reproduce their definitions here for completeness.

\section{Conclusions}

In this work, we presented an ACL2 mechanization of a generic framework for weak memory, as well as a novel proof of an established result for this framework. We hope to incorporate this framework into our ongoing research into how a theorem prover like ACL2 can be used to verify correctness properties of real-world concurrent programs. Our most immediate future work consists of applying these concepts (actually, a simplification of these concepts) to proofs on a multi-processor x86 model, but this work suggests the possibility of applying a general weak memory framework to other models as well.

This work was supported by NSF. We gratefully acknowledge the many helpful comments and discussions provided by Jade Alglave and Matt Kaufmann.

\bibliographystyle{eptcs}
\bibliography{axiomatic_weak_memory_bib}

\begin{thebibliography}{10}
\providecommand{\bibitemdeclare}[2]{}
\providecommand{\surnamestart}{}
\providecommand{\surnameend}{}
\providecommand{\urlprefix}{Available at }
\providecommand{\url}[1]{\texttt{#1}}
\providecommand{\href}[2]{\texttt{#2}}
\providecommand{\urlalt}[2]{\href{#1}{#2}}
\providecommand{\doi}[1]{doi:\urlalt{http://dx.doi.org/#1}{#1}}
\providecommand{\bibinfo}[2]{#2}

\bibitemdeclare{phdthesis}{alglave_thesis}
\bibitem{alglave_thesis}
\bibinfo{author}{Jade \surnamestart Alglave\surnameend} (\bibinfo{year}{2010}):
  \emph{\bibinfo{title}{A Shared Memory Poetics}}.
\newblock Ph.D. thesis, \bibinfo{school}{Universit\'{e} Paris 7}.

\bibitemdeclare{article}{alglave_cats}
\bibitem{alglave_cats}
\bibinfo{author}{Jade \surnamestart Alglave\surnameend}, \bibinfo{author}{Luc
  \surnamestart Maranget\surnameend} \& \bibinfo{author}{Michael \surnamestart
  Tautschnig\surnameend} (\bibinfo{year}{2014}): \emph{\bibinfo{title}{Herding
  Cats - Modelling, simulation, testing, and data-mining for weak memory.}}
\newblock {\sl \bibinfo{journal}{TOPLAS (to appear)}}.
\newblock \urlprefix\url{http://arxiv.org/abs/1308.6810}.

\bibitemdeclare{article}{boudol09}
\bibitem{boudol09}
\bibinfo{author}{G{\'e}rard \surnamestart Boudol\surnameend} \&
  \bibinfo{author}{Gustavo \surnamestart Petri\surnameend}
  (\bibinfo{year}{2009}): \emph{\bibinfo{title}{Relaxed Memory Models: An
  Operational Approach}}.
\newblock {\sl \bibinfo{journal}{SIGPLAN Not.}}
  \bibinfo{volume}{44}(\bibinfo{number}{1}), pp. \bibinfo{pages}{392--403},
  \doi{10.1145/1594834.1480930}.

\bibitemdeclare{inproceedings}{chong08}
\bibitem{chong08}
\bibinfo{author}{Nathan \surnamestart Chong\surnameend} \&
  \bibinfo{author}{Samin \surnamestart Ishtiaq\surnameend}
  (\bibinfo{year}{2008}): \emph{\bibinfo{title}{Reasoning About the ARM Weakly
  Consistent Memory Model}}.
\newblock In: {\sl \bibinfo{booktitle}{Proceedings of the 2008 ACM SIGPLAN
  Workshop on Memory Systems Performance and Correctness: Held in Conjunction
  with the Thirteenth International Conference on Architectural Support for
  Programming Languages and Operating Systems (ASPLOS '08)}},
  \bibinfo{series}{MSPC '08}, \bibinfo{publisher}{ACM}, \bibinfo{address}{New
  York, NY, USA}, pp. \bibinfo{pages}{16--19}, \doi{10.1145/1353522.1353528}.

\bibitemdeclare{article}{hoare69}
\bibitem{hoare69}
\bibinfo{author}{C.~A.~R. \surnamestart Hoare\surnameend}
  (\bibinfo{year}{1969}): \emph{\bibinfo{title}{An Axiomatic Basis for Computer
  Programming}}.
\newblock {\sl \bibinfo{journal}{Commun. ACM}}
  \bibinfo{volume}{12}(\bibinfo{number}{10}), pp. \bibinfo{pages}{576--580},
  \doi{10.1145/363235.363259}.

\bibitemdeclare{article}{lamport79}
\bibitem{lamport79}
\bibinfo{author}{L.~\surnamestart Lamport\surnameend} (\bibinfo{year}{1979}):
  \emph{\bibinfo{title}{How to Make a Multiprocessor Computer That Correctly
  Executes Multiprocess Programs}}.
\newblock {\sl \bibinfo{journal}{IEEE Trans. Comput.}}
  \bibinfo{volume}{28}(\bibinfo{number}{9}), pp. \bibinfo{pages}{690--691},
  \doi{10.1109/TC.1979.1675439}.

\bibitemdeclare{inproceedings}{owens09}
\bibitem{owens09}
\bibinfo{author}{Scott \surnamestart Owens\surnameend}, \bibinfo{author}{Susmit
  \surnamestart Sarkar\surnameend} \& \bibinfo{author}{Peter \surnamestart
  Sewell\surnameend} (\bibinfo{year}{2009}): \emph{\bibinfo{title}{A Better x86
  Memory Model: X86-TSO}}.
\newblock In: {\sl \bibinfo{booktitle}{Proceedings of the 22Nd International
  Conference on Theorem Proving in Higher Order Logics}},
  \bibinfo{series}{TPHOLs '09}, \bibinfo{publisher}{Springer-Verlag},
  \bibinfo{address}{Berlin, Heidelberg}, pp. \bibinfo{pages}{391--407},
  \doi{10.1007/978-3-642-03359-9\_27}.

\bibitemdeclare{article}{owicki76}
\bibitem{owicki76}
\bibinfo{author}{Susan \surnamestart Owicki\surnameend} \&
  \bibinfo{author}{David \surnamestart Gries\surnameend}
  (\bibinfo{year}{1976}): \emph{\bibinfo{title}{An Axiomatic Proof Technique
  for Parallel Programs}}.
\newblock {\sl \bibinfo{journal}{Acta Informatica}} \bibinfo{volume}{6}, pp.
  \bibinfo{pages}{319--340}, \doi{10.1007/BF00268134}.

\bibitemdeclare{article}{sarkar11}
\bibitem{sarkar11}
\bibinfo{author}{Susmit \surnamestart Sarkar\surnameend},
  \bibinfo{author}{Peter \surnamestart Sewell\surnameend},
  \bibinfo{author}{Jade \surnamestart Alglave\surnameend}, \bibinfo{author}{Luc
  \surnamestart Maranget\surnameend} \& \bibinfo{author}{Derek \surnamestart
  Williams\surnameend} (\bibinfo{year}{2011}):
  \emph{\bibinfo{title}{Understanding POWER Multiprocessors}}.
\newblock {\sl \bibinfo{journal}{SIGPLAN Not.}}
  \bibinfo{volume}{46}(\bibinfo{number}{6}), pp. \bibinfo{pages}{175--186},
  \doi{10.1145/1993316.1993520}.

\bibitemdeclare{article}{sarkar09}
\bibitem{sarkar09}
\bibinfo{author}{Susmit \surnamestart Sarkar\surnameend},
  \bibinfo{author}{Peter \surnamestart Sewell\surnameend},
  \bibinfo{author}{Francesco~Zappa \surnamestart Nardelli\surnameend},
  \bibinfo{author}{Scott \surnamestart Owens\surnameend}, \bibinfo{author}{Tom
  \surnamestart Ridge\surnameend}, \bibinfo{author}{Thomas \surnamestart
  Braibant\surnameend}, \bibinfo{author}{Magnus~O. \surnamestart
  Myreen\surnameend} \& \bibinfo{author}{Jade \surnamestart Alglave\surnameend}
  (\bibinfo{year}{2009}): \emph{\bibinfo{title}{The Semantics of x86-CC
  Multiprocessor Machine Code}}.
\newblock {\sl \bibinfo{journal}{SIGPLAN Not.}}
  \bibinfo{volume}{44}(\bibinfo{number}{1}), pp. \bibinfo{pages}{379--391},
  \doi{10.1145/1594834.1480929}.

\end{thebibliography}

\end{document}